\documentclass[11pt]{article}
\usepackage[margin=1in]{geometry}
\usepackage{amsmath,amssymb,amsthm,mathtools}
\usepackage{enumitem}
\usepackage{booktabs}
\usepackage[hidelinks]{hyperref}
\usepackage[capitalise,nameinlink]{cleveref}

\newtheorem{theorem}{Theorem}
\newtheorem{proposition}[theorem]{Proposition}

\newtheorem{corollary}[theorem]{Corollary}
\newtheorem{fact}{Fact}

\theoremstyle{definition}
\newtheorem{definition}[theorem]{Definition}
\newtheorem{example}[theorem]{Example}
\newtheorem{conjecture}[theorem]{Conjecture}
\theoremstyle{remark}
\newtheorem{remark}[theorem]{Remark}
\crefname{fact}{Fact}{Facts}

\newcommand{\C}{\mathbb{C}}
\newcommand{\N}{\mathbb{N}}
\newcommand{\Hol}{\operatorname{Holant}}
\newcommand{\Tr}{\operatorname{Tr}}
\newcommand{\soc}{\operatorname{soc}}
\newcommand{\G}{\mathcal{G}}
\newcommand{\cross}{\mathbin{\bowtie}}

\title{Generalized Fibonacci gates are Frobenius algebras:\\
a proof of Liu's conjecture\\[4pt]
\large Summary of results}
\author{Aaron Sterling\\ \small Thistleseeds \quad \texttt{asterling@thistleseeds.com}}
\date{Draft, October 2026}

\begin{document}
\maketitle

\begin{abstract}
Liu introduced generalized Fibonacci gates on domains of size 3 and 4, proved that
Holant problems built from them are tractable by a cut-and-merge algorithm, and
conjectured that Fibonacci gates with matching polynomial-time algorithms exist on all
higher domains. We extend the definition uniformly to every domain size $k$ and show
that a family of domain-$k$ generalized Fibonacci gates is the same thing as a
$k$-dimensional commutative Frobenius (equivalently, Artinian Gorenstein) algebra $A$
equipped with an orthonormal basis containing $1$. Under this dictionary Liu's
parameter constraints are exactly associativity of $A$, a gate is an element $w\in A$,
merging two gates is multiplication in $A$, and the Holant value of a connected graph
$G$ is $\tau\bigl(h^{\,|E|-|V|+1}\prod_v w_v\bigr)$ for a fixed ``handle'' element
$h\in A$. This proves Liu's conjecture for every $k$, with an algorithm whose cost is
polynomial in $k$ as well as in the graph. We also give a closure criterion for linear recurrences of
any depth (a balance condition implies closure, and is necessary when the unit is a
combination of the colours), and we map the landscape: semisimple algebras are exactly
Liu's orthogonal frames, smoothable algebras with nilpotents are collided frames, and families whose algebra is not smoothable (``frameless'' families)
exist exactly when $k\ge 14$. Finally we show that every non-semisimple part of a
family is invisible to the Holant on connected graphs with two or more independent cycles.
This note states the results and sketches the proofs; the
proofs of the main theorems are short and are given essentially in full.
\end{abstract}

\section{Introduction}

\paragraph{Background.}
On the Boolean domain, Cai, Lu and Xia~\cite{CLX08} introduced \emph{Fibonacci gates}:
symmetric signatures $[f_0,\dots,f_n]$ with $f_{j+2}=m f_{j+1}+f_j$. Holant problems over
Fibonacci gates with a common parameter $m$ are solvable in polynomial time by a direct
combinatorial algorithm. Liu~\cite{Liu24} extended the definition to domains of size $3$
and $4$ (his Definitions~1 and~2), by asking each ``medium-sized'' depth-$2$ triangle
(tetrahedron) of signature values to satisfy a linear relation whose parameters obey
polynomial constraints. He proved tractability by cut-and-merge (his Theorems~1 and~2)
and wrote:
\begin{quote}
\emph{``I conjecture that there are Fibonacci Gates and the associated P-time algorithms
on all higher domains.''}
\end{quote}

\paragraph{What this note does.}
\begin{enumerate}[leftmargin=*,label=(\arabic*)]
\item We give the natural definition of generalized Fibonacci gates on any domain
$[k]$ (\cref{def:liu-k}). For $k=3$ it is Liu's Definition~1 verbatim. For $k=4$ it is
Liu's Definition~2 with his relation~(8) included among the constraints; his proof of
Theorem~2 uses~(8) (see \cref{rem:eq8}).
\item \textbf{Dictionary} (\cref{thm:dictionary}). Domain-$k$ generalized Fibonacci
families are in bijection with triples $(A,\tau,\varphi)$, where $A$ is a
$k$-dimensional commutative Frobenius algebra with Frobenius form $\tau$ and $\varphi$ is
a $\tau$-orthonormal basis of $A$ with $\varphi_1=1$. Liu's recurrence coefficients are the structure constants of $A$, and his
constraints are associativity.
\item \textbf{Merge is multiplication} (\cref{thm:merge}). Gates are elements of $A$;
cross-component merge is the product in $A$; same-component merge is multiplication by
a fixed element $h$. The Holant has a closed form (\cref{cor:holant}). This settles
Liu's conjecture for all $k$ (\cref{cor:liu}).
\item \textbf{Generalized Liu Conjecture} (our generalization; \cref{thm:glc}). For
linear-recurrence classes of any depth, a single balance condition implies closure under
cut-and-merge, and it is necessary whenever $1\in\operatorname{span}(\varphi_c)$ (as in
Liu's normalization). When the merge pairing is nondegenerate, balance forces Gorenstein
and depth $\le2$, so
Liu's depth-$2$ definition already captures every nondegenerate case.
\item \textbf{Landscape} (\cref{sec:landscape}). Semisimple $A$ gives exactly the
orthogonal-frame gates that Liu starts from. Non-reduced smoothable $A$ gives collided
frames.
Non-smoothable $A$ gives frameless families, which exist if and only if $k\ge14$
(using Casnati--Jelisiejew--Notari, Iarrobino--Emsalem, and
Biaggi--Draisma--Seynnaeve). On connected graphs with at least two independent cycles,
only the semisimple part of $A$ contributes to the Holant (\cref{prop:vanishing}).
\end{enumerate}

\Cref{tab:dictionary} summarizes the dictionary.

\begin{table}[ht]
\centering
\begin{tabular}{@{}ll@{}}
\toprule
\textbf{Holant / Liu side} & \textbf{Algebra side} \\
\midrule
domain size $k$ & $\dim A$ \\
first colour (Red) & unit $1\in A$ \\
recurrence parameters ($s,x,y,t$ at $k=3$) & structure constants $\tau(\varphi_l\varphi_m\varphi_n)$ \\
apex coefficients $1$ (doubled) / $0$ (mixed) & orthonormality of the basis $\varphi$ \\
Liu's parameter constraints & associativity of $A$ \\
a gate of arity $n$ & an element $w\in A$, via $F^w_n$ \\
the $k$ ``top'' values $g(1^{n-1}c)$ & the coordinates of $w$ in the basis $\varphi$ \\
cross-component merge & multiplication $w\cdot z$ \\
same-component merge & multiplication by $h=\sum_c\varphi_c^2$ \\
Holant of a connected graph & $\tau(h^{\,|E|-|V|+1}\prod_v w_v)$ \\
orthogonal frame (Liu \S2, Valiant) & semisimple $A\cong\C^k$ \\
collision of frame vectors & nilpotents (non-reduced, smoothable $A$) \\
frameless family & non-smoothable $A$; requires $k\ge14$ \\
\bottomrule
\end{tabular}
\caption{The dictionary. Here $\tau$ is the Frobenius form of $A$.}
\label{tab:dictionary}
\end{table}

\section{Holant problems, signatures and merges}\label{sec:holant}

Fix a domain $[k]=\{1,\dots,k\}$ and work over $\C$. A symmetric signature $f$ of arity $n$
is a function $[k]^n\to\C$ invariant under permuting its arguments. Equivalently, it is a
function $f(\gamma)$ of the colour-count vector $\gamma\in\N^k$, $|\gamma|=n$, where
$\gamma_c$ is the number of inputs coloured $c$. A signature grid assigns a signature
$f_v$ of arity $\deg v$ to each vertex of a multigraph $G=(V,E)$ (loops allowed), and
\[
\Hol(G)=\sum_{\sigma:E\to[k]}\ \prod_{v\in V} f_v\bigl(\sigma|_{E(v)}\bigr).
\]
Liu's algorithm cuts every edge into two dangling half-edges and then reconnects them one
at a time. There are two operations~\cite[Figs.~1--2]{Liu24}:
\begin{itemize}[leftmargin=*]
\item \emph{Cross-component merge.} For gates $F$ (arity $n$) and $G$ (arity $m$),
\[(F\cross G)(x_1,\dots,x_{n-1},y_1,\dots,y_{m-1})=\sum_{c\in[k]}F(x_1,\dots,x_{n-1},c)\,
G(y_1,\dots,y_{m-1},c).\]
\item \emph{Same-component merge.} For a gate $F$ of arity $n$,
$F^{\circ}(x_1,\dots,x_{n-2})=\sum_{c\in[k]}F(x_1,\dots,x_{n-2},c,c)$.
\end{itemize}
After all edges are reconnected we have an arity-$0$ gate, and its value is $\Hol(G)$. A
class of symmetric signatures is \emph{closed under cut-and-merge} if both operations keep
it inside the class. Closure, together with a fixed-size description of each gate,
gives a polynomial-time algorithm.

\section{Generalized Fibonacci gates on every domain}\label{sec:liu}

\subsection{Liu's definition, and its uniform extension}
Write a domain-$k$ signature as $g(\gamma)$ with $\gamma=(\gamma_1,\dots,\gamma_k)$, so
$\gamma_1$ counts the first colour. Liu's domain-3 relations~\cite[Eq.~(3)]{Liu24} read
\begin{align*}
g(\gamma-2e_1+2e_2)&=g(\gamma)+s\,g(\gamma-e_1+e_2)+x\,g(\gamma-e_1+e_3),\\
g(\gamma-2e_1+e_2+e_3)&=\phantom{g(\gamma)+{}}x\,g(\gamma-e_1+e_2)+y\,g(\gamma-e_1+e_3),\\
g(\gamma-2e_1+2e_3)&=g(\gamma)+y\,g(\gamma-e_1+e_2)+t\,g(\gamma-e_1+e_3),
\end{align*}
with $sy+xt+1=x^2+y^2$. Two features generalize. The ``apex'' coefficient is $1$ for a
doubled colour and $0$ for mixed colours. The bottom coefficients are shared in a
totally symmetric pattern: $T_{222}=s$, $T_{223}=x$, $T_{233}=y$, $T_{333}=t$, where
$T_{mm'l}$ is the coefficient of $g(\gamma-e_1+e_l)$ in the relation for
$g(\gamma-2e_1+e_m+e_{m'})$. Liu's domain-4 relations~(6) have the same two features,
with ten parameters forming a symmetric tensor on $\{2,3,4\}$.

\begin{definition}[generalized Fibonacci gate, domain $k$]\label{def:liu-k}
Let $T=(T_{lmn})_{l,m,n=2}^{k}$ be a totally symmetric array of complex numbers. A
symmetric signature $g$ on $[k]$ of arity $\ge2$ is a \emph{generalized Fibonacci gate
with parameter $T$} if, for all $\gamma$ with $\gamma_1\ge2$ and all $m,m'\in\{2,\dots,k\}$,
\begin{equation}\label{eq:rec}
g(\gamma-2e_1+e_m+e_{m'})=\delta_{mm'}\,g(\gamma)+\sum_{l=2}^{k}T_{mm'l}\;g(\gamma-e_1+e_l).
\end{equation}
Every unary signature counts as a gate. The parameter $T$ is \emph{admissible} if, for
all $m,n,p,q\in\{2,\dots,k\}$,
\begin{equation}\label{eq:assoc}
\delta_{mn}\delta_{pq}+\sum_{l=2}^{k}T_{mnl}T_{lpq}
\;=\;\delta_{pn}\delta_{mq}+\sum_{l=2}^{k}T_{pnl}T_{lmq}.
\end{equation}
A set of signatures is \emph{generalized Fibonacci} if all of its members are gates with
one common admissible parameter.
\end{definition}

Relation~\eqref{eq:assoc} contains both of Liu's ``dot product'' rules, the one for his
quadrics in~(7)~\cite[p.~9]{Liu24} and the one for~(8)~\cite[p.~10]{Liu24}, written for
all $k$. The $\delta$'s are the apex entries $1/0$.

\begin{example}[$k=3$]
With $(m,n,p,q)=(2,2,3,3)$, \eqref{eq:assoc} reads $1+T_{222}T_{233}+T_{223}T_{333}
=T_{322}T_{223}+T_{323}T_{323}$, that is, $sy+xt+1=x^2+y^2$. This is Liu's constraint,
and every other instance of~\eqref{eq:assoc} at $k=3$ is trivial or equivalent to it.
\end{example}

\begin{remark}[$k=4$ and Liu's Equation~(8)]\label{rem:eq8}
A Gr\"obner-basis computation over $\mathbb{Q}[a,\dots,p]$ shows that, at $k=4$, the
ideal generated by the admissibility equations~\eqref{eq:assoc} \emph{equals} the ideal
generated by Liu's Equations~(7) and~(8) together. Equation~(8) is \emph{not} a
consequence of~(7), not even set-theoretically. The rational point
$(a,b,c,d,e,f,h,i,j,p)=(1,0,1,0,0,0,0,0,1,1)$ satisfies all four equations of~(7)
(including the cubic) but violates the first equation of~(8). At this point, merging
two generic binary gates in the sense of Definition~2 gives a signature that is not
even symmetric, so closure fails for Definition~2 as written. Liu's proof of Theorem~2
does invoke~(8) (``by Equations 7 and 8''), so it establishes Theorem~2 once~(8) is
added to Definition~2. Under our dictionary, (7)\,+\,(8) is exactly associativity.
\end{remark}

\subsection{Liu's conjecture, in precise form}
We read Liu's conjecture as follows.

\begin{conjecture}[Liu's Conjecture, precise form]\label{conj:liu}
For every $k\ge2$ there are generalized Fibonacci gates on $[k]$ (in the sense of
\cref{def:liu-k}). For every finite generalized Fibonacci set $\mathcal F$,
$\Hol(\mathcal F)$ is computable in polynomial time by cut-and-merge.
\end{conjecture}

Existence of \emph{orthogonal-frame} signatures on every domain was never in doubt.
They are orthogonal holographic transforms of weighted equalities
(\cite{Val08}; \cite[\S2]{Liu24}). The content of the conjecture is that the
\emph{combinatorial} definition, a recurrence plus polynomial constraints, extends to all
$k$ and stays closed under merging. As we will see, the class it defines is strictly
larger than the frame signatures.

\section{Frobenius (Artinian Gorenstein) algebras: the facts we use}\label{sec:facts}

All algebras are commutative, unital, finite-dimensional over $\C$. For $a\in A$ let
$L_a$ denote multiplication by $a$. The main theorems use only
\cref{f:structure,f:frobenius,f:casimir,f:orthonormal}, which are standard.
\Cref{f:macaulay,f:smooth} are used only in \cref{sec:landscape}.

\begin{fact}[Structure]\label{f:structure}
$A\cong\prod_i A_i$ with each $A_i$ local Artinian (via the primitive idempotents).
$A$ is reduced (no nonzero nilpotents) iff $A\cong\C^r$ (semisimple).
\end{fact}

\begin{fact}[Frobenius $\Leftrightarrow$ Gorenstein]\label{f:frobenius}
The following are equivalent: (a) there is a linear $\tau:A\to\C$ such that
$(x,y)\mapsto\tau(xy)$ is nondegenerate (a \emph{Frobenius form}); (b) $A\cong A^{*}$
as $A$-modules; (c) every local factor $A_i$ has a one-dimensional socle (Artinian
Gorenstein). If $\tau_0$ is one Frobenius form, the Frobenius forms are exactly
$\tau_0(u\,\cdot)$ with $u\in A^\times$. A form $\tau$ is Frobenius iff $\ker\tau$
contains no nonzero ideal. See~\cite[\S21.2]{Eis95}, \cite[\S2.2]{Kock04}.
\end{fact}

\begin{fact}[Copairing / Casimir element]\label{f:casimir}
For a Frobenius form $\tau$ with dual bases $\{x_\alpha\},\{x^\alpha\}$
($\tau(x_\alpha x^\beta)=\delta_{\alpha\beta}$), the element
$\omega_\tau=\sum_\alpha x_\alpha\otimes x^\alpha\in A\otimes A$ is independent of the
basis and satisfies, for all $a,x\in A$:
\begin{gather*}
\text{(i) } (a\otimes1)\,\omega_\tau=(1\otimes a)\,\omega_\tau;\qquad
\text{(ii) } \textstyle\sum_\alpha\tau(x\,x_\alpha)\,x^\alpha=x;\\
\text{(iii) } \tau(h_\tau a)=\Tr(L_a),\ \text{ where } h_\tau=\textstyle\sum_\alpha x_\alpha x^\alpha.
\end{gather*}
The element $h_\tau$ is the \emph{handle element} of 2D TQFT~\cite{Abr96,Kock04}.
\end{fact}

\begin{fact}[Orthonormal bases]\label{f:orthonormal}
Every nondegenerate symmetric bilinear form on a finite-dimensional $\C$-space has an
orthonormal basis, and every vector $v$ with $B(v,v)=1$ extends to one (by Gram--Schmidt,
avoiding isotropic pivots; square roots exist in $\C$). Over a field without enough
square roots one gets only an orthogonal basis, with $B\cong\operatorname{diag}(d_1,\dots,d_k)$.
That case is handled by weighted equality edges (\cref{rem:fields}).
\end{fact}

\begin{fact}[Macaulay inverse systems]\label{f:macaulay}
For a polynomial $F$ in $N$ variables, the algebra
$A_F=\C[\partial_1,\dots,\partial_N]/\operatorname{Ann}(F)$ is local Artinian Gorenstein, and every local Artinian Gorenstein algebra arises this
way~\cite{IK99}. For a form $F$ of degree $s$, $A_F$ is graded with socle in degree $s$
and symmetric Hilbert function.
\end{fact}

\begin{fact}[Smoothability]\label{f:smooth}
$A$ is \emph{smoothable} if it is a flat limit of reduced algebras $\C^{\dim A}$. Every
Gorenstein algebra of length $\le13$ is smoothable, since the Gorenstein locus of the
Hilbert scheme of $\le13$ points is irreducible~\cite{CJN15}. A general local
Gorenstein algebra with Hilbert function $(1,6,6,1)$ (length $14$) is not
smoothable~\cite{IE78}. See~\cite{Jel23} for a survey.
\end{fact}

\section{Main theorems}\label{sec:main}

\begin{definition}[Fibonacci datum]\label{def:datum}
A \emph{Fibonacci datum on $[k]$} is a triple $(A,\tau,\varphi)$: an algebra $A$, a
Frobenius form $\tau$, and $\varphi=(\varphi_1,\dots,\varphi_k)\in A^k$ such that
\begin{equation}\label{eq:casimir}
\sum_{c=1}^{k}\varphi_c\otimes\varphi_c=\omega_\tau,
\qquad\text{equivalently}\qquad
\sum_{c=1}^{k}\tau(x\varphi_c)\,\varphi_c=x\quad\text{for all }x\in A.
\end{equation}
Its \emph{gates} are the symmetric signatures
\[
F^w_n(c_1,\dots,c_n)=\tau\bigl(w\,\varphi_{c_1}\varphi_{c_2}\cdots\varphi_{c_n}\bigr),\qquad
w\in A,\ n\ge0.
\]
Put $h=\sum_c\varphi_c^2$. By~\eqref{eq:casimir} and \cref{f:casimir}, $h=h_\tau$.
\end{definition}

If $k=\dim A$, condition~\eqref{eq:casimir} says exactly that $\varphi$ is a
$\tau$-orthonormal basis: its dual basis is itself. In general $\eqref{eq:casimir}$
forces $\dim A\le k$, because the left side has tensor rank $\le k$ while
$\omega_\tau$ has rank $\dim A$.

\begin{theorem}[Merge is multiplication]\label{thm:merge}
For any Fibonacci datum and all $w,z\in A$:
\[
F^w_n\cross F^z_m=F^{wz}_{n+m-2},\qquad (F^w_n)^{\circ}=F^{wh}_{n-2}.
\]
In particular the gates of a datum form a class closed under cut-and-merge, and each
gate stays described by a single element of $A$.
\end{theorem}

\begin{proof}
Write $\varphi_{\mathbf a}=\varphi_{a_1}\cdots\varphi_{a_{n-1}}$ for the free inputs
$\mathbf a$ of the first gate and $\varphi_{\mathbf b}$ likewise. Then
\begin{align*}
(F^w_n\cross F^z_m)(\mathbf a,\mathbf b)&=\sum_c\tau(w\varphi_{\mathbf a}\varphi_c)\,\tau(z\varphi_{\mathbf b}\varphi_c)
=\tau\Bigl(z\varphi_{\mathbf b}\sum_c\tau(w\varphi_{\mathbf a}\varphi_c)\varphi_c\Bigr)\\
&=\tau(z\varphi_{\mathbf b}\,w\varphi_{\mathbf a})=F^{wz}_{n+m-2}(\mathbf a,\mathbf b),
\end{align*}
using~\eqref{eq:casimir} with $x=w\varphi_{\mathbf a}$. The same-component case is immediate:
$\sum_c\tau(w\varphi_{\mathbf a}\varphi_c\varphi_c)=\tau(wh\,\varphi_{\mathbf a})$. Symmetry of the merged
signature in all $n+m-2$ inputs, the first step of Liu's proofs, is just commutativity
of $A$.
\end{proof}

\begin{corollary}[Closed form for the Holant]\label{cor:holant}
Let $G$ be connected, with gate $F^{w_v}_{\deg v}$ at each vertex $v$, and let
$\beta(G)=|E|-|V|+1$ be its cyclomatic number. Then
\[
\Hol(G)=\tau\Bigl(h^{\beta(G)}\prod_{v\in V}w_v\Bigr).
\]
For disconnected $G$, multiply over components.
\end{corollary}

\begin{proof}
Merge along a spanning tree ($|V|-1$ cross merges), then close the remaining $\beta(G)$
edges (same-component merges), and apply \cref{thm:merge}. The arity-$0$ gate $F^W_0$ has
value $\tau(W)$.
\end{proof}

\begin{theorem}[Dictionary]\label{thm:dictionary}
Fix $k\ge2$. The following are in natural bijection:
\begin{enumerate}[label=(\alph*)]
\item admissible parameters $T$ (\cref{def:liu-k});
\item isomorphism classes of triples $(A,\tau,\varphi)$ with $\dim A=k$, $\tau$ a Frobenius
form, and $\varphi$ a $\tau$-orthonormal basis with $\varphi_1=1$, i.e.\ Fibonacci data
on $[k]$ in which the first colour is the unit.
\end{enumerate}
The bijection is $T_{lmn}=\tau(\varphi_l\varphi_m\varphi_n)$. Under it, for every arity
$n\ge1$, the generalized Fibonacci gates with parameter $T$ are exactly the gates
$F^w_n$ of the datum, and $w=\sum_{c}g(1^{n-1}c)\,\varphi_c$ is read off from the $k$
``top'' values. In particular every algebra in~(b) is Artinian Gorenstein.
\end{theorem}

\begin{proof}[Proof sketch]
(b)$\Rightarrow$(a). Expanding a product in the orthonormal basis gives
\[\varphi_m\varphi_{m'}=\sum_c\tau(\varphi_m\varphi_{m'}\varphi_c)\varphi_c
=\delta_{mm'}\cdot1+\sum_{l\ge2}T_{mm'l}\varphi_l.\] Here $T$ is totally symmetric because
$A$ is commutative, and the coefficient of $1$ is
$\tau(\varphi_m\varphi_{m'})=\delta_{mm'}$. So
$F^w_n(\gamma-2e_1+e_m+e_{m'})=\tau(w\varphi^{\gamma-2e_1}\varphi_m\varphi_{m'})$ satisfies
\eqref{eq:rec}, using $\varphi_1=1$. Expanding
$(\varphi_m\varphi_n)\varphi_p=\varphi_m(\varphi_n\varphi_p)$ and comparing coefficients
of $\varphi_q$ gives~\eqref{eq:assoc}.

(a)$\Rightarrow$(b). Given admissible $T$, define a product on $\C^k$ with basis
$e_1,\dots,e_k$ by $e_1$ a unit and
$e_me_{m'}=\delta_{mm'}e_1+\sum_lT_{mm'l}e_l$. It is commutative because $T$ is
symmetric. Comparing coefficients shows that associativity on basis triples is
exactly~\eqref{eq:assoc}; the coefficient of $e_1$ matches automatically by symmetry of
$T$. Set $\tau(e_1)=1$ and $\tau(e_m)=0$ for $m\ge2$. Then
$\tau(e_me_{m'})=\delta_{mm'}$, so the basis is orthonormal and $\tau$ is Frobenius.

Gates. Repeated use of~\eqref{eq:rec} expresses every value $g(\gamma)$ through the $k$
values with at most one non-first colour, by induction on $|\gamma|-\gamma_1$. So a
gate of arity $n\ge1$ is determined by its top values. $F^w_n$ is a gate with top values
$\tau(w\varphi_c)$, and these range over all of $\C^k$ as $w$ ranges over $A$ because
$\tau$ is nondegenerate. Hence the gates are exactly the $F^w_n$.
\end{proof}

\begin{corollary}[Liu's Conjecture]\label{cor:liu}
\Cref{conj:liu} holds for every $k\ge2$. Generalized Fibonacci gates exist on every
domain, and for a finite generalized Fibonacci set with parameter $T$, $\Hol(G)$ is
computed by: build the multiplication table of $A$ from $T$; read off
$w_v=\sum_c f_v(1^{\deg v-1}c)\,\varphi_c$; output $\tau(h^{\beta}\prod_vw_v)$ per
component. This uses $O\bigl(k^3(|V|+\log|E|)\bigr)$ arithmetic operations after the
input is read, polynomial in $k$ as well as in $G$.
\end{corollary}

For existence, any $(A,\tau)$ with $\tau(1)=1$ and $\dim A=k$ works by
\cref{f:orthonormal}, for example $A=\C^k$ with
$\tau(e_i)=t_i\ne0$ and $\sum t_i=1$. These frame families form a component of dimension
$k(k-1)/2$, which matches the degree-of-freedom counts in~\cite[\S\S2--3]{Liu24}.
Non-frame families exist for every $k\ge2$: for example $A=\C[t]/(t^2)\times\C^{k-2}$,
with $\tau$ normalized so that $\tau(1)=1$ (\cref{ex:boolean} is the case $k=2$).

\subsection{Generalized Liu Conjecture: recurrences of any depth}

Liu's relations have depth~$2$: each value is expressed through values with one or two
fewer first-colour inputs. A natural generalization allows recurrences of any depth.
The algebra handles this uniformly.

Let $S=\C[y_1,\dots,y_k]$, let $I\subset S$ be an ideal with $A=S/I$
finite-dimensional, and let $\varphi_c$ be the image of $y_c$. The \emph{recurrence class}
$\G(I)$ consists of the signatures
$F^\lambda_n(c_1,\dots,c_n)=\lambda(\varphi_{c_1}\cdots\varphi_{c_n})$ for
$\lambda\in A^{*}$ and $n\ge0$. Every $P\in I$ gives a linear relation satisfied by all
of them. It relates values of different arities when $P$ is not homogeneous; with
$\varphi_1=1$, as in Liu's normalization, one homogenizes by $y_1$ and gets relations
within a single arity. The \emph{depth} of $\G(I)$ is the least $d$ such that $A$ is spanned by
monomials of degree $<d$ in the $\varphi_c$. Liu's classes have depth $2$. Put
$\Omega=\sum_c\varphi_c\otimes\varphi_c$, and call $\G(I)$ \emph{balanced} if
$(a\otimes1)\Omega=(1\otimes a)\Omega$ for all $a\in A$.

\begin{theorem}[Closure criterion; our Generalized Liu Conjecture]\label{thm:glc}
\leavevmode
\begin{enumerate}[label=(\roman*)]
\item If $\G(I)$ is balanced, it is closed under cut-and-merge, at every depth. Indeed
$F^\lambda_n\cross F^\mu_m=F^{\lambda\star\mu}_{n+m-2}$ and
$(F^\lambda_n)^\circ=F^{\lambda(h\,\cdot)}_{n-2}$, where
$(\lambda\star\mu)(z)=(\lambda\otimes\mu)\bigl((z\otimes1)\Omega\bigr)$ is a commutative,
associative product on $A^{*}$. So $\Hol(\G(I))$ is polynomial-time.
\item Conversely, if $1\in\operatorname{span}(\varphi_c)$ and $\G(I)$ is closed under
cross-component merges at all sufficiently large arities, then $\G(I)$ is balanced.
\item If $\G(I)$ is balanced and $\Omega$ is nondegenerate (as a bilinear form on
$A^*$), then $A$ is Gorenstein, $\tau:=(\Omega^\#)^{-1}(1)$ is a Frobenius form,
$(A,\tau,\varphi)$ is a Fibonacci datum, $\lambda\star\mu$ corresponds to
multiplication, and the depth is at most $2$.
\end{enumerate}
\end{theorem}

\begin{proof}[Proof sketch]
(i) For monomials $x,y$, $(F^\lambda\cross F^\mu)(x,y)=(\lambda\otimes\mu)((x\otimes
y)\Omega)=(\lambda\otimes\mu)((xy\otimes1)\Omega)$ by balance. Commutativity and
associativity of $\star$ follow from two more applications of balance.
(ii) If $1\in\operatorname{span}(\varphi)$, monomials of each large degree span $A$, so
$H(x,y)=\sum_c\lambda(x\varphi_c)\mu(y\varphi_c)$ is a bilinear form on $A$. Closure
says $H(x,y)=\nu(xy)$, hence $H(x,y)=H(xy,1)$ for all $\lambda,\mu$. Taking $x=1$ gives
$\sum_c\varphi_c\otimes y\varphi_c=\sum_c y\varphi_c\otimes\varphi_c$.
(iii) Balance says that $\Omega^\#:A^*\to A$, $\lambda\mapsto\sum_c\lambda(\varphi_c)\varphi_c$,
is $A$-linear. If it is bijective then $A^*\cong A$ (so $A$ is Gorenstein by
\cref{f:frobenius}), and with $\tau=(\Omega^\#)^{-1}(1)$ we get
$\Omega^\#(\tau(w\,\cdot))=w$, which is~\eqref{eq:casimir}. Finally,
$A=\operatorname{im}\Omega^\#\subseteq\operatorname{span}(\varphi_c)$, so the depth is
at most $2$.
\end{proof}

\begin{remark}[What part~(iii) says about depth]
Every balanced recurrence class of depth $\ge3$ (in particular, every closed one with
$1\in\operatorname{span}(\varphi_c)$) is \emph{degenerate}: its merges factor
through the proper ideal $\operatorname{im}\Omega^\#$ and lose information. Gates that
look ``deeper'' because frame vectors have collided (osculating, nilpotent behaviour)
are \emph{not} deeper recurrences. They are depth-$2$ gates whose algebra $A$ is
non-reduced (\cref{sec:landscape}). In this sense Liu's depth-$2$ definition already
covers every nondegenerate case.
\end{remark}

\begin{remark}[Other fields; weighted edges]\label{rem:fields}
The results of \cref{sec:main} work over any field of characteristic $\ne2$ that
contains the needed square roots. (\Cref{f:structure}, \cref{prop:frames} and
\cref{thm:frameless} below use that $\C$ is algebraically closed.) Over $\mathbb{Q}$ a Frobenius pairing may admit only an orthogonal basis with
$\tau(\varphi_c^2)=d_c$. Then \cref{thm:merge} holds verbatim for the Holant with binary
weighted-equality edges $\operatorname{diag}(d_c^{-1})$, and the sum-of-squares
(Hasse--Minkowski) obstruction to an \emph{orthonormal} basis concerns only unweighted
edges over non-closed fields.
\end{remark}

\section{The landscape: frames, collisions, frameless gates}\label{sec:landscape}

\begin{proposition}[Semisimple $=$ orthogonal frames]\label{prop:frames}
A Fibonacci datum has $A$ reduced iff its gates are orthogonal frame sums
$F_n=\sum_{i=1}^r a_i\,v_i^{\otimes n}$ with $\langle v_i,v_j\rangle=\delta_{ij}N_i$ and
$N_i\ne0$. Explicitly, for $A=\C^r$ with idempotents $e_i$, $\varphi_c=\sum_iv_i(c)e_i$,
$N_i=1/\tau(e_i)$ and $a_i=\tau(e_i)w_i$. Merge becomes $a_ib_iN_i$: Liu's
\S2--3 starting point, and the orthogonal holographic transformation to weighted
equality~\cite{Val08}.
\end{proposition}

\begin{proof}[Proof sketch]
($\Rightarrow$) $\omega_\tau=\sum_ie_i\otimes e_i/\tau(e_i)$, while
$\sum_c\varphi_c\otimes\varphi_c=\sum_{i,j}\langle v_i,v_j\rangle e_i\otimes e_j$.
Compare the two. ($\Leftarrow$) Suppose every gate $F^w_n$ is a frame sum
$\sum_ia_i(w)v_i^{\otimes n}$ over one fixed orthogonal, non-isotropic family
$v_1,\dots,v_r$. Cross merge multiplies coefficients as $a_ib_iN_i$, so
$w\mapsto(a_i(w)N_i)_i$ is an algebra map $A\to\C^r$. It is injective: if all
$a_i(w)=0$ then $F^w_1=0$, i.e.\ $\tau(w\varphi_c)=0$ for all $c$, and $w=0$
by~\eqref{eq:casimir}. A subalgebra of $\C^r$ is reduced.
\end{proof}

\begin{example}[Collision on the Boolean domain]\label{ex:boolean}
Take $A=\C[t]/(t^2)$ with $\tau(a+bt)=a+b$, so that $\tau(1)=1$. The pair $\varphi_1=1$,
$\varphi_2=i(1-t)$ is orthonormal: $\tau(\varphi_2^2)=-\tau(1-2t)=1$ and
$\tau(\varphi_2)=0$. The gates have $f_j=i^j\tau\bigl(w(1-jt)\bigr)$, an affine function
of $j$ times $i^j$. This is the span of $(1,i)^{\otimes n}$ and its first-order
deformation. They satisfy the Fibonacci recurrence of~\cite{CLX08} with $m=2i$, namely
$f_{j+2}=2if_{j+1}+f_j$, whose characteristic polynomial $(x-i)^2$ has a double root: a
collided frame. Here $h=2t$ is nilpotent. In arity $\ge3$ these are vanishing
signatures in the sense of Cai--Guo--Williams~\cite{CGW16}, and \cref{prop:vanishing}
is an analogue of that phenomenon for arbitrary generalized Fibonacci families. Their
vanishing class is broader in one way, since it vanishes on every graph, and narrower
in another, since it is a Boolean-domain class.
\end{example}

\begin{proposition}[Non-semisimple parts vanish on graphs with $\beta\ge2$]\label{prop:vanishing}
Write $A=\prod_iA_i$ with local factors (\cref{f:structure}). On each non-field factor
$A_i$, the handle element satisfies $h_i\in\soc(A_i)$, so $h_i^2=0$. Consequently, for
connected $G$ with $\beta(G)\ge2$,
\[
\Hol(G)=\tau\Bigl(e_{\mathrm{red}}\,h^{\beta(G)}\prod_vw_v\Bigr),
\]
where $e_{\mathrm{red}}$ is the sum of the idempotents of the field factors $A_i\cong\C$.
In other words, on such graphs every generalized Fibonacci family computes the same
Holant as its frame part (\cref{prop:frames}). If $A$ is local and $A\ne\C$, the
Holant vanishes on every connected graph with $\beta\ge2$.
\end{proposition}

\begin{proof}[Proof sketch]
By \cref{f:casimir}(iii), for $a\in A_i$ we have $\tau(h_ia)=\Tr(L_a|_{A_i})=(\dim
A_i)\,\varepsilon_i(a)$, where $\varepsilon_i$ is the residue map. Hence $\tau(h_i\,m_ix)=0$
for all $x$, and by nondegeneracy $h_i\,m_i=0$. So $h_i$ lies in the socle, which is
contained in $m_i$ when $A_i\ne\C$. Therefore $h_i^2=0$. Apply \cref{cor:holant}.
\end{proof}

\begin{definition}
A Fibonacci datum (or generalized Fibonacci family) is \emph{frameless} if $A$ is not
smoothable. Then the family is not a flat limit of orthogonal-frame families, not
even of collided ones.
\end{definition}

\begin{theorem}[Frameless families exist exactly from domain $14$]\label{thm:frameless}
\leavevmode
\begin{enumerate}[label=(\roman*)]
\item If a Fibonacci datum on $[k]$ is frameless, then $k\ge\dim A\ge14$.
\item For every $k\ge14$ there is a frameless generalized Fibonacci family on $[k]$ in
Liu's normal form (\cref{def:liu-k}).
\item Equivalently, in tensor language: the admissible parameters $T$ of
\cref{def:liu-k}, extended by the unit, are symmetric tensors satisfying Robeva's
quadratic equations for odeco tensors on the unital locus. For $k\le13$ they all lie
in the closure of the odeco tensors, and for $k\ge14$ some do not~\cite[Cor.~2.6.5]{BDS22}.
\end{enumerate}
\end{theorem}

\begin{proof}[Proof sketch]
(i) $\dim A\le k$ by the rank remark after \cref{def:datum}, and $\dim A\ge14$ by
\cref{f:smooth} (CJN). (ii) Let $A_0$ be a general $(1,6,6,1)$ algebra (\cref{f:macaulay}),
non-smoothable by~\cite{IE78}. Let $s$ span the socle $(A_0)_3$, and let $\tau_0$ be
the projection onto $(A_0)_3=\C s$, zero in degrees $\le2$. For $c\in\C$,
$\tau_0((1+cs)\,\cdot)$ is a Frobenius form on $A_0$ (\cref{f:frobenius}) taking the
value $c$ at $1$. On $A=A_0\times\C^{k-14}$, which is still non-smoothable, set
$\tau=\tau_0((1+cs)\,\cdot)$ on $A_0$ and $\tau(e_i)=1$ on the field factors, with
$c=15-k$. Then $\tau$ is Frobenius and $\tau(1)=1$. Extend $1$ to a $\tau$-orthonormal
basis by \cref{f:orthonormal}. (iii) Robeva's
equations are associativity of $(x,y)\mapsto T(x,y,\cdot)$~\cite{Rob16,BDHR17}, which is
\eqref{eq:assoc} with the unit adjoined. The identification with the Gorenstein locus,
and the thresholds, are due to Biaggi--Draisma--Seynnaeve~\cite{BDS22}, who build on
CJN.
\end{proof}

\begin{remark}[Interpretation]
\Cref{thm:frameless} shows that the gate class is genuinely larger than (limits of)
holographic transforms of equality. \Cref{prop:vanishing} limits what this buys for
complexity: frameless and other non-semisimple contributions are visible only on
forests and unicyclic graphs, or in settings that leave the single family $\G$, such as
mixing with other signatures. The tensor-side phenomenon is due to
BDS~\cite{BDS22}. What is new here is its reading as Holant gates and the Holant formula.
\end{remark}

\section{Status and verification}

\begin{itemize}[leftmargin=*]
\item \textbf{Complete proofs here.} \Cref{thm:merge,thm:dictionary,thm:glc} and
\cref{cor:holant,cor:liu,prop:frames,prop:vanishing}. They use only
\cref{f:structure,f:frobenius,f:casimir,f:orthonormal}.
\item \textbf{Cited.} The thresholds in \cref{thm:frameless} rest on~\cite{CJN15,IE78}
and their translation by~\cite{BDS22}.
\item \textbf{Machine checks.} $k=3$: admissibility is $sy+xt+1=x^2+y^2$ (by hand and
symbolically). $k=4$: the Gr\"obner-basis identification of~\eqref{eq:assoc} with
Liu's (7)\,+\,(8), and the independence of~(8), as in \cref{rem:eq8}. $k=14$: an explicit
frameless family built from a random $(1,6,6,1)$ algebra. Brute-force Holant values on
small graphs match \cref{cor:holant}: relative error $\le 10^{-9}$ in double precision
on graphs with $\beta=0,1$. At $\beta\ge2$ both sides vanish numerically, as
\cref{prop:vanishing} predicts for this local algebra.
\item \textbf{For the full version.} Liu's ``reducible'' cases (a zero frame vector,
domain-separable signatures), which correspond to data where colour~$1$ is not the unit
and are covered by \cref{thm:merge}; bit complexity; non-symmetric signatures;
Holant$^*$. Since every unary signature is a gate when $\varphi$ is a basis,
Holant$^*$ over a generalized Fibonacci set is also tractable.
\end{itemize}

\end{document}